\documentclass{article}
\usepackage{mlspconf}

\usepackage{appendix}
\usepackage{makecell}

\usepackage[english]{babel}

\usepackage{ifpdf}

\usepackage{cite} 
\usepackage{url}
\usepackage{hyperref}

\usepackage[pdftex]{graphicx}
\graphicspath{{./figures/}}
\usepackage{color}
\usepackage{pgf, tikz, pgfplots}
\usetikzlibrary{shapes, arrows, automata, plotmarks}
\usetikzlibrary{calc,hobby,decorations}

\usepackage[cmex10]{amsmath}
\usepackage{amsfonts, amssymb, amsthm}
\usepackage{mathrsfs}

\usepackage{algorithm,algpseudocode}
\algnewcommand{\LeftComment}[1]{\Statex \(\triangleright\) #1}

\usepackage{color, colortbl}
\definecolor{Gray}{gray}{0.9}
\usepackage{array}

\usepackage{enumitem}

\usepackage[margin=15mm]{geometry}
\usepackage{multirow}
\usepackage{rotating}
\usepackage{subcaption}
\input{pennColors.sty}
\usepackage{needspace}

\input{mySymbol.sty}

\def\Tr{\mathsf{T}}

\newtheorem{proposition}{\hspace{0pt}\bf Proposition}

\newtheorem{definition}{\hspace{0pt}\bf Definition}

\newcommand{\QED}{\hfill\ensuremath{\blacksquare}}

\title{GRAPH-ADAPTIVE ACTIVATION FUNCTIONS FOR GRAPH NEURAL NETWORKS}
\name{Bianca Iancu$^{*}$, Luana Ruiz$^{\dagger}$,  Alejandro Ribeiro$^{\dagger}$ and Elvin Isufi$^{*}$}
\address{$^{*} $ Intelligent Systems Department, Delft University of Technology, Delft, The Netherlands\\ $\dagger $ Department of Electrical and Systems Engineering, University of Pennsylvania, Philadelphia, USA \\ e-mails: bianca.iancu026@gmail.com, \{rubruiz, aribeiro\}@seas.upenn.edu, E.Isufi-1@tudelft.nl}
\begin{document}
\ninept
\maketitle
\begin{abstract}
Activation functions are crucial in graph neural networks (GNNs) as they allow defining a nonlinear family of functions to capture the relationship between the input graph data and their representations. This paper proposes activation functions for GNNs that not only adapt to the graph into the nonlinearity, but are also distributable. To incorporate the feature-topology coupling into all GNN components, nodal features are nonlinearized and combined with a set of trainable parameters in a form akin to graph convolutions. The latter leads to a graph-adaptive trainable nonlinear component of the GNN that can be implemented directly or via kernel transformations, therefore, enriching the class of functions to represent the network data. Whether in the direct or kernel form, we show permutation equivariance is always preserved. We also prove the subclass of graph-adaptive \emph{max} activation functions are Lipschitz stable to input perturbations. Numerical experiments with distributed source localization, finite-time consensus, distributed regression, and recommender systems corroborate our findings and show improved performance compared with pointwise as well as state-of-the-art localized nonlinearities.

\end{abstract}

\section{Introduction}
\label{sec:intro}
\vskip-.2cm

Graph neural networks (GNNs) are parametric architectures suitable for learning a nonlinear mapping for data defined over graphs such as social, sensor, and biological network data \cite{wu2020comprehensive,gama2020graphs}. By interweaving graph filters with pointwise nonlinearities, GNNs express the function map in a layered form and learn compositions of features that account for the data-topology coupling \cite{defferrard17-cnngraphs,gama18-gnnarchit}. Another property GNNs inherit from graph filters is the distributed implementation \cite{shuman2018distributed,isufi2016autoregressive,segarra17-linear}. Distributed computation facilitates scalability of computation and endows the system with robustness to failures of the processing unit. The latter is fundamental in applications involving consensus, optimization, and control \cite{sandryhaila2014finite, tsitsiklis1986distributed, jadbabaie2003coordination}.

Building on spectral graph theory, \cite{bruna14-deepspectralnetworks} defined graph convolutional neural networks by multiplying feature representations in the Laplacian eigenspace with trainable kernels. Subsequently, \cite{defferrard17-cnngraphs} used finite impulse response (FIR) graph filters to combine features in the vertex domain by means of a polynomial in the Laplacian matrix. The work in \cite{gama18-gnnarchit} follows the same idea but builds a polynomial filter in any graph representation matrix (e.g., adjacency, Laplacian). Differently from \cite{bruna14-deepspectralnetworks}, \cite{defferrard17-cnngraphs,gama18-gnnarchit} are also readily distributable architectures with appropriate choices of graph pooling (i.e., not altering the graph structure; e.g., zero-padding) and with pointwise activation functions. On the other hand, \cite{bianchi2019graph} builds a GNN with distributable autoregressive moving average graph filters \cite{isufi2016autoregressive}, which capture a broader family of functions at the expense of computation cost. Parallel to these efforts, \cite{Velickovic18-GraphAttentionNetworks} proposes attention-like mechanisms to adapt the edge weights to the task at hand. More recently, the work in \cite{Isufi20-EdgeNets} showed that all the above architectures are equivalent and fall under the framework of edge varying GNN (EdgeNet). Altogether, these works capture the data-graph coupling only linearly through graph filters, while they ignore the coupling in the nonlinear pointwise component (e.g., ReLU). To improve the representation power of GNNs, \cite{ruiz18-local} proposed localized activation functions that account for the graph topology by operating on node neighborhoods of different resolutions. However, the latter accounts only for the graph and not the data-topology coupling, since it ignores the edge weights and the data propagation between neighbors. Localized activation functions are also not distributable beyond the one-hop neighborhood, hence missing multi-hop information between nodes.

To address these limitations, we put forward a new family of activation functions that adapt to the data-topology coupling in the surrounding of a node. The nodal features obtained from graph filtering are shifted prior to local-nonlinearization in a form akin to graph convolutions. These nonlinear features are subsequently combined with a set of trainable parameters to accordingly weigh the information at different neighborhood resolutions. The resolution radius is a design parameter and allows adapting the GNN nonlinear component to the task at hand. Besides being graph-adaptive and distributable, these activation functions preserve two properties of theoretical interest for GNNs, namely permutation equivariance and Lipschitz stability to perturbations \cite{gama19-stability}. Concretely, our contribution is threefold.
\begin{enumerate}[noitemsep,nolistsep]
\item We develop a new family of nonlinearities for GNNs that are graph-adaptive to the surrounding of a node and distributable. The first class [Def.~\ref{def:distr_activ_fct}] nonlinearizes shifted features in the surrounding of a node in their direct form. The second class [Def.~\ref{def:distr_kernel_activ_fct}] transforms the shifted features with graph-adaptive kernels prior to nonlinearization.
\item We prove that: $(a)$ the proposed nonlinearities are permutation equivariant [Prop.~\ref{prop:perm_equiv}], i.e. the output of the respective GNN architecture is agnostic to node labeling; $(b)$ the \emph{max} graph-adaptive nonlinearity is Lipschitz stable to input perturbations [Prop.~\ref{thm:lipschitz}].
\item We propose distributed GNN tasks with graph-adaptive nonlinearities for source localization, finite-time consensus, signal denoising, and rating prediction in recommender systems.
\end{enumerate}


\section{Graph Neural Networks}
\label{sec:gnns}
\vskip-.2cm

Consider a graph $\mathcal{G} = (\mathcal{V}, \mathcal{E})$ with vertex set $\mathcal{V}$ of cardinality $|\mathcal{V}| = N$ and edge set $\mathcal{E} \subseteq \mathcal{V} \times \mathcal{V}$ of cardinality $|\mathcal{E}| = M$. An edge is a tuple $e_{ij} = (i, j)$ connecting nodes $i$ and $j$. The neighborhood of node $i$ is the set of nodes $\mathcal{N}_i = \{j | (i,j) \in \mathcal{E}\}$ connected to $i$. Associated to $\mathcal{G}$ is the graph shift operator (GSO) matrix $\bbS \in \reals^{N \times N}$, whose sparsity pattern matches the graph structure. That is, entry $(i,j)$ satisfies $[\bbS]_{i,j} = s_{i,j} \neq 0$ only if $i = j$ or $(i,j) \in \mathcal{E}$. Commonly used GSOs include the adjacency matrix, the graph Laplacian, and their normalized and translated forms.

On the vertices of $\ccalG$, we define a graph signal $\bbx \in \mathbb{R}^N$ whose $i$th component is the value at node $i$. We consider applications where graph signals are processed in a \textit{distributed} fashion. A typical example is in sensor networks without access to a centralized processing unit and where each sensor communicates only with its neighbor sensors. 

\noindent \textbf{Graph convolution.} A graph convolution is defined as a graph filter $\bbH(\bbS)$ that can be written as a polynomial of the GSO $\bbS$ \cite{segarra17-linear}. For an input signal $\bbx$ and filter coefficients $\bbh = [h_0, \ldots, h_K]^\top$, the output $\bby \in \mathbb{R}^N$ of the graph convolutional filter is computed as
\begin{equation} \label{eqn:graph_conv}
    \bby = \bbH(\bbS) \bbx = \sum_{k=0}^{K}h_k \bbS^k \bbx.
\end{equation}
Due to the locality of $\bbS$, graph convolutions can be run distributively. When building the output $\bby$, we need to compute the terms $\bbS\bbx, \ldots, \bbS^K\bbx$. Since $\bbS$ is local, operation $\bbS\bbx$ requires one-hop node exchanges and so, by writing $\bbS^k\bbx = \bbS(\bbS^{k-1}\bbx) = \bbS\bbx^{(k-1)}$, node $i$ can compute signal $\bbx^{(k)}$ through exchange of previous shifted information $\bbx^{(k-1)}$ with its neighbors. This recursion allows for distributed communications and computational cost of order $\ccalO(MK)$, while the trainable parameters defining (\ref{eqn:graph_conv}) are of order $\ccalO(K)$ \cite{segarra17-linear}. 

\noindent \textbf{Graph convolutional neural networks (GCNNs).} We consider a GCNN of $L$ graph convolutional layers followed by a shared fully connected layer per node. Each convolutional layer comprises a bank of graph filters [cf. \eqref{eqn:graph_conv}] and a nonlinearity. At layer $l$, the GCNN takes as input $F_{l-1}$ features $\{\bbx_{l-1}^g\}_{g = 1}^{F_{l-1}}$ from layer $(l-1)$ and produces $F_l$ output features $\{\bbx_{l}^f\}_{f = 1}^F$. Each input feature $\bbx_{l-1}^g$ is processed by a parallel bank of $F_l$ graph filters $\{\bbH_{l}^{fg}(\bbS)\}_f$. The filter outputs are aggregated over the input index $g$ to yield the $f$th convolved feature \vskip-.45cm
\begin{equation}\label{eq.interm_linear}
\bbz_l^f \!=\! \sum_{g = 1}^{F_{l-1}}\bbH_{l}^{fg}(\bbS)\bbx_{l-1}^g \!=\! \sum_{g = 1}^{F_{l-1}}\sum_{k = 0}^Kh_{kl}^{fg}\bbS^k\bbx_{l-1}^g,~\for~ f = 1, \ldots, F_{l}.
\end{equation}
The convolved feature $\bbz_l^f$ is subsequently passed through an activation function $\sigma(\cdot)$ to obtain the $f$th convolutional layer output \vskip-.2cm
\begin{equation}\label{eq.activ_fct}
\bbx_l^f = \sigma(\bbz_l^f), \quad\for~ f = 1, \ldots, F.
\end{equation}

The output features of the last convolutional layer $L$, $\bbx_{L}^1, \ldots, \bbx_{L}^{F_L}$, represent the final convolutional features. These features are interpreted as a collection of $F_L$ graph signals, where on node $i$ we have the $F_L \times 1$ feature vector $\bbchi_{Li} = [x_{Li}^1, \ldots, x_{Li}^{F_L}]^\top$. Each node locally combines the features $\bbchi_{Li}$ with a one-layer perceptron to obtain the output \vskip-.2cm
\begin{equation}
\tby = \bbH_{\text{FC}}\bbchi_{Li}
\end{equation}
where matrix $\bbH_{\text{FC}} \in \mathbb{R}^{F_\text{o}\times F_L}$ maps the $F_L$ convolutional features to the $F_o$ output features (e.g., the number of classes). The parameters in $\bbH_{\text{FC}}$ are shared among nodes to keep the number of trainable parameters independent of the graph dimensions (i.e., $N$ and $M$), but only dependent on the filter order and the number of features and layers.

By grouping all learnable parameters into the set $\ccalH = \{\bbh_l^{fg}; \bbH_{\text{FC}}\}_{lfg}$, we can consider the GCNN as a map $\bbPhi(\cdot)$ that takes as input a graph signal $\bbx$, a GSO $\bbS$, and a set of parameters $\ccalH$ to produce the output \vskip-.2cm
\begin{equation}\label{eq.GCNNout}
\bbPhi(\bbx; \bbS; \ccalH):= \tby.
\end{equation}
The output \eqref{eq.GCNNout} is computed for a training set $\ccalT = \{(\bbx, \bby)\}$ of $|\ccalT|$ pairs, where $\bby$ are the target representations. 


\textbf{Activation functions.} The activation function $\sigma(\cdot)$ in (\ref{eq.activ_fct}) can be any of the conventional pointwise activation functions, such as ReLU ($\sigma(x) = \max(0, x)$), or a localized activation function \cite{ruiz18-local}. Differently from the pointwise, localized activation functions consider the features at the neighborhood of each node $i$ in the nonlinear GCNN component \cite{ruiz18-local}. For a graph signal feature $\bbx$ the localized activation function is based on two local operators, namely:
\begin{itemize} [noitemsep,nolistsep]
    \item \emph{local max operator}, $\text{max}(\bbS, \bbx)$, whose output is a graph signal $\bbz$ with $i$th entry being the maximum value of the signal in the neighborhood, i.e., $z_i = [\text{max}(\bbS, \bbx)]_i=  \max\big(\{x_j:v_j \in \ccalN_i    \}    \big)$;
    \item \emph{local median operator}, $\text{med}(\bbS, \bbx)$, whose output is a graph signal $\bbz$ with $i$th entry being the median value of the signal in the neighborhood, i.e., $z_i = [\text{med}(\bbS, \bbx)]_i=  \text{med}\big(\{x_j:v_j \in \ccalN_i    \}    \big)$.
\end{itemize}
For simplicity, we denote both local operators with the generic local function $f(\bbS, \bbx)$. Then, the localized activation function is defined as \vskip-.2cm
\begin{equation}\label{eq.localizedactiv}
\sigma(\bbx) = \beta \text{ReLU}(\bbx) +  \sum_{k=1}^{K} h_{\sigma k} f(\textbf{S}^k, \bbx).
\end{equation}
where $f(\bbS^k,\bbx)$ applies the local activation function to the signal values of the $k$-hop neighbors and parameters $\beta$ and $\bbh_\sigma \!=\! [h_{\sigma 1}, \ldots,  h_{\sigma K}]^\top$ are \textit{learned} \cite{ruiz18-local}. A GCNN with localized activation functions can thus be written as the map $\bbPhi(\bbx; \bbS; \ccalH)$ with parameters $\ccalH \!=\! \{\bbh_l^{fg};\bbh_{\sigma l}^f; \bbH_{\text{FC}}\}_{lfg}$. 

As it follows from \eqref{eq.localizedactiv}, localized activation functions ignore the edge weights and require information from the non-immediate $k$-hop neighbors, which makes them not distributable. Hence, in distributed settings, the order $K$ in \eqref{eq.localizedactiv} is limited to one. To address this limitation, we propose two new activation functions based on local operators and kernel functions to account for the graph structure and be distributable.                 

\section{Graph-Adaptive Activation Functions}
\label{sec:dln}
\vskip-.2cm

In this section, we first define the  \textit{graph-adaptive localized activation functions}, which are based on arbitrary nonlinear operators acting on the one-hop neighborhood of a node (Section \ref{sbs:dlaf}). Then, we define the \textit{graph-adaptive kernel activation functions} (Section \ref{sbs:dkaf}). Finally, we prove the proposed nonlinearities are permutation equivariant and stable to input perturbations (Section \ref{sbs:properties}).

\subsection{Graph-Adaptive Localized Activation Functions} \label{sbs:dlaf}
To start, let us first define the basic building block for graph-adaptive activation functions: the \textit{shifted localized operator} (SLO).
\begin{definition}[Shifted Localized Operator]  \label{def:shifted_loc_op}
Let $\ccalG$ be an $N$-node graph with shift operator $\bbS$, $\bbx$ a signal, and $\bbS^k \bbx$ the $k$th shifted signal. Consider an arbitrary nonlinear localized function $f(\cdot,\ccalN_i): \reals^N \to \reals^N$, which at node $i$ computes the local nonlinear operation $[f(\bbx,\ccalN_i)]_i = f(\{x_j\}_{j \in \ccalN_i})$. For this choice of $f(\cdot,\ccalN_i)$, the $k$-hop shifted localized operator maps input $\bbx$ to output $\bbz \in \reals^N$ as \vskip-.2cm
\begin{equation} \label{eqn:shifted_localized_operator}
    z_i = [f(\bbS^{k}\bbx, \mathcal{N}_i)]_i = f(\{[\bbS^k\bbx]_j: j \in \ccalN_i\})\text{.}
\end{equation}
\end{definition}
That is, the SLO shifts the signal $k$ times to obtain $\bbS^k \bbx$, and then replaces the value of this signal at each node $i$ by a nonlinear aggregation $f(\cdot, \ccalN_i)$ of the signal values within the one-hop neighborhood of $i$.
The SLO utilizes information locally available at each node to account for the signal-topology coupling for nodes that are $k$-hops away. We can now define \textit{graph-adaptive nonlinear graph filters} as follows.

\begin{definition}[Shifted Localized Graph Filter] \label{def:shifted_loc_graph_filter}
Consider the shifted localized operator induced by an arbitrary nonlinear localized function $f(\cdot,\ccalN_i)$ [cf. Def. \ref{def:shifted_loc_op}], and let $\bbh_{\sigma} = [h_{\sigma 1}, \ldots, h_{\sigma K}]^\top$ be a vector of parameters. The output of the shifted localized graph filter applied to signal $\bbx$, w.r.t. the shift operator $\bbS$, is the signal $\bbz \in \reals^N$ with $i$th entry \vskip-.2cm
\begin{equation} \label{eqn:shifted_loc_graph_filter}
    z_i = \sum_{k=1}^{K} h_{\sigma k} [f(\bbS^{k}\bbx, \mathcal{N}_i)]_i\text{.}
\end{equation}
\end{definition}

Definition \ref{def:shifted_loc_graph_filter} implies the output of a shifted localized graph filter is a linear combination of the SLOs $f(\bbS^k\bbx, \ccalN_i)$ at different resolutions. Hence, shifted localized graph filters inherit the localization property of SLOs, as they incorporate the graph structure up to $K$ hops away accessing only neighboring information. These nonlinear filters can be employed to define \textit{graph-adaptive localized activation functions}.
\begin{definition}[Graph-Adaptive Localized Activation Function] \label{def:distr_activ_fct}
Consider a scalar $\beta$ and vector $\bbh_\sigma \! = \! [h_{\sigma l1}^f, \ldots, h_{\sigma lK}^f]^\top$ of learnable parameters. At layer $l$, the graph-adaptive localized activation function maps the linear features $\bbz_l^f$ [cf. \ref{eq.interm_linear}] to the output features $\bbx_l^f$ following the recursion \vskip-.3cm
\begin{equation}\label{eqn:distr_activ_fct}
    [\bbx_l^f]_i = \beta \text{ReLU}([\bbz_l^f]_i) + \sum_{k=1}^{K} h_{\sigma lk}^{f} [f(\bbS^{k}\bbz_l^f, \mathcal{N}_i)]_i\text{.}
\end{equation}
\end{definition}

Definition \ref{def:distr_activ_fct} combines the pointwise ReLU nonlinearity and the shifted localized graph filters [cf. Def. \ref{def:shifted_loc_graph_filter}] into a single graph-adaptive localized nonlinearity for GNNs. The latter is distributable and localized because, even though the resolution ---given by the shift order--- can be arbitrarily large, the SLO $f(\cdot, \ccalN_i)$ [cf. Def. \ref{def:shifted_loc_op}] operates only in the one-hop neighborhood. In Section 4, we evaluate this activation function for $f(\cdot, \ccalN_i)$ being the max and median, leading to the \textit{graph-adaptive max and median activation function}, respectively.

\subsection{Graph-Adaptive Kernel Activation Functions} \label{sbs:dkaf}

The graph-adaptive kernel activation functions replace the localized nonlinear function $f(\cdot, \ccalN_i)$ by a localized \textit{kernel} to enrich the representation power.
Let $\bbx_i^{(k)} \in \reals^{|\ccalN_i|}$ denote the vector containing $|\ccalN_i|$ copies of the $k$ shifted signal at node $i$, $[\bbS^k\bbx]_i$, i.e. $\bbx_i^{(k)} = \bbone_{|\ccalN_i|} \otimes [\bbS^k\bbx]_i$ where $\bbone_{|\ccalN_i|}$ is the vector of ones of dimension $|\ccalN_i|$ and $\otimes$ is the Kronecker operator. Additionally, consider the vector containing the values at neighbors $j \in \ccalN_i$ of the $k$th shifted signal $\bbS^k\bbx$, i.e. $\bbx^{(k)}_{j \in \ccalN_i} = [\bbS^k \bbx]_{j \in \ccalN_i}$.
With this notation in place, we define a graph kernel operator as follows.

\begin{definition}[Kernel Operator] \label{def:kernel_op}
Let $\ccalG$ be an $N$-node graph with shift operator $\bbS$, $\bbx$ a signal, and $\bbS^k \bbx$ the $k$th shifted signal. Consider an arbitrary kernel function $g(\cdot,\ccalN_i): \reals^{|\ccalN_i|} \to \reals^{|\ccalN_i|}$, which at node $i$ computes the nonlinear local operation $[g(\bbx,\ccalN_i)]_i = g(\widetilde{\bbx}_i, \bbx_{j \in \ccalN_i})$, where $\widetilde{\bbx}_i = \bbone_{|\ccalN_i|} \otimes [\bbx]_i$ is a vector of dimensionality $|\ccalN_i|$ containing copies of signal $\bbx$ at node $i$. The $k$-hop shifted kernel operator mapping from $\bbx$ to $\bbz \in \reals^N$ has the entries \vskip-.2cm
\begin{equation} \label{eqn:kernel_operator}
    z_i = [g(\bbS^k\bbx, \ccalN_i)]_i := g(\bbx^{(k)}_i, \bbx^{(k)}_{j \in \ccalN_i}).
    \end{equation}
\end{definition}
Definition \ref{def:kernel_op} shows the kernel operator first shifts the input signal $\bbx$ as $\bbS^k \bbx$ and then replaces the signal value at each $i$ by the kernel value $g(\cdot,\ccalN_i)$ in the one-hop neighborhood of $i$. Thus, the kernel operator employs only local information at each node to account for the signal-topology coupling up to $k$-hops away from a node. For the kernel function $g(\cdot, \ccalN_i)$ we will employ the Gaussian kernel 
\begin{equation} \label{eqn:gaussian_kernel}
         g(x,y) = \exp\left(-{||\bbx-\bby||^2}/{2\gamma^2}\right),
    \end{equation}
where scalar $\gamma$ is tunable. We can now define \textit{kernel graph filters}.

\begin{definition}[Kernel Graph Filter] \label{def:kernel_graph_filter}
Consider a kernel operator [cf. $\ref{def:kernel_op}$] with kernel function $g(\cdot, \ccalN_i)$ and let $\bbh_{\sigma} = [h_{\sigma 1},\ldots, h_{\sigma K}]^\top$ be a vector of parameters. The output of the kernel graph filter applied to signal $\bbx$, w.r.t. the shift operator $\bbS$, is the signal $\bbz \in \reals^N$ with $i$th entry \vskip-.2cm
\begin{equation} \label{eqn:kernel_graph_filter}
    z_i = \sum_{k=1}^{K} h_{\sigma k} [g(\bbS^{k}\bbx, \ccalN_i)]_i.
\end{equation}
\end{definition}
Definition \ref{def:kernel_graph_filter} implies the output of the kernel graph filter is a linear combination of the kernel operator applied to each $k$-shifted signal $\bbS^k \bbx$ for $1 \leq k \leq K$. Kernel graph filters thus preserve the localization properties of kernel operators, i.e., they account for the topology of the graph up to $K$-hops away accessing only information in the one-hop neighborhood.
These kernel graph filters can be further employed to define the \textit{graph-adaptive kernel activation function} as follows.

\begin{definition}[Graph-Adaptive Kernel Activation Function] \label{def:distr_kernel_activ_fct} 
Consider a scalar $\beta$ and vector $\bbh_\sigma \! = \! [h_{\sigma l1}^f, ..., h_{\sigma lK}^f]^\top$ of learnable parameters. At layer $l$, the graph-adaptive kernel activation function maps the linear features $\bbz_l^f$ [cf. \ref{eq.interm_linear}] to the output features $\bbx_l^f$ following the recursion \vskip-.2cm
\begin{equation} \label{eqn:distr_kernel_activ_fct}
    [\bbx_l^f]_i = \beta \text{ReLU}([\bbz_l^f]_i) + \sum_{k=1}^{K} h_{\sigma lk}^f [g(\bbS^{k}\bbz_l^f, \mathcal{N}_i)]_i\text{.}
\end{equation}
\end{definition}
Definition \ref{def:distr_kernel_activ_fct} combines the pointwise ReLU and kernel graph filters [cf. Def. \ref{def:kernel_graph_filter}] into a single graph-adaptive kernel activation function. This activation function is distributable and localized because, even though the resolution ---given by the shift order--- can be arbitrarily large, the kernel $g(\cdot, \ccalN_i)$ operates only in the one-hop neighborhood.

\bigbreak
In both proposed activation functions, coefficients $\{\beta, \bbh_\sigma\}$ are trainable, meaning these nonlinearities adapt the multi-hop resolution weights to the task at hand. Because these coefficients are shared among nodes, the number of parameters to learn for a graph-adaptive activation function is independent of the graph size. This allows GCNNs to \textit{scale}. Note that even though the nonlinear functions $f(\cdot, \ccalN_i)$ or the kernel functions $g(\cdot, \ccalN_i)$ act only on the one-hop neighborhood, they are applied to the shifted signals $\bbS^k\bbx$, therefore they account for the feature-graph coupling (up to $K$-hops away) in a nonlinear fashion. This is an advantage over traditional GCNNs with pointwise nonlinearities, in which the graph topology is only incorporated through linear encodings generated by graph convolutions.

Definitions \ref{def:distr_activ_fct} and \ref{def:distr_kernel_activ_fct} implement \textit{fully graph-adaptive} GCNNs that, at each layer, apply a graph convolution followed by a graph-adaptive activation function. The distributed GCNN is given by the map \vskip-.2cm
\begin{equation} \label{eqn:mapping_distr_GNN}
    \bbPhi(\bbx; \bbS, \ccalH, \ccalW) :=   \tby.
\end{equation}
The GCNN output now depends on both the coefficients $\mathcal{H}$ [cf. \eqref{eq.GCNNout}] and on the nonlinear activation functions coefficients $\ccalW = \{\bbh_{\sigma l}^f\}_{lf} \cup \{\beta\}$. 

\subsection{Properties of Graph-Adaptive Nonlinearities} \label{sbs:properties}

A key property GCNNs with pointwise activation functions inherit from graph convolutions is \textit{permutation equivariance} \cite{ruiz18-local}. The output of a GCNN is invariant to node relabeling and, more importantly, GCNNs exploit graph symmetries to generalize learned representations to different graph signals that share some of these symmetries. Herein, we show that permutation equivariance also holds for graph-adaptive nonlinearities. We will also discuss a property that is specific to the graph-adaptive localized max activation: Lipschitz stability to input perturbations.  

\smallskip \noindent \textbf{Permutation equivariance.} Consider the graph convolutional filter $\bbH(\bbS)$ [cf. \eqref{eqn:graph_conv}] and let $\bbP$ be an $N \times N$ permutation matrix satisfying $\bbP^\Tr \bbP = \bbP \bbP^\Tr = \bbI$. If we permute the GSO $\bbS$ and input $\bbx$ respectively as $\bbS' = \bbP^\Tr\bbS\bbP$ and $\bbx'=\bbP^\Tr\bbx$, we get the corresponding graph convolution output  
\begin{equation} \label{eqn:graph_conv_perm_equiv}
\bby' = \bbH(\bbS')\bbx' = \bbH(\bbP^\Tr\bbS\bbP)\bbP^\Tr\bbx = \bbP^\Tr\bbH(\bbS)\bbP\bbP^\Tr\bbx = \bbP^\Tr\bby\ .
\end{equation}
Because pointwise activation functions are scalar and by definition permutation equivariant, \eqref{eqn:graph_conv_perm_equiv} implies GCNNs with pointwise nonlinearities are invariant to node relabelings. For GCNNs with graph-adaptive activation functions, it is then desirable to retain this property. This is guaranteed by the following proposition.


\begin{proposition}[Permutation equivariance] \label{prop:perm_equiv}
Consider a graph signal $\bbx$ defined on an $N$-node graph $\mathcal{G}$ with GSO $\bbS$ . Let $\bbPhi(\bbx; \bbS, \mathcal{H}, \mathcal{W})$ be the output of a GCNN with graph-adaptive activation functions [cf. \eqref{eqn:mapping_distr_GNN}] and let $\bbP$ be an $N \times N$ permutation matrix. The GNN $\bbPhi(\bbx; \bbS, \mathcal{H}, \mathcal{W})$ satisfies \vskip-.2cm
\begin{equation}\label{eqn:perm_equivariance_distr_kernel}
    \bbPhi(\bbP^\top \bbx; \bbP^\top \bbS \bbP, \ccalH, \ccalW) = \bbP^\top \bbPhi(\bbx; \bbS, \ccalH, \ccalW)
\end{equation}
i.e., GNNs with graph-adaptive activation functions are permutation equivariant.
\end{proposition} 
\begin{proof}
For the proof, we refer the reader to the Appendix.
\end{proof}

\vspace{1ex}
\noindent \textbf{Lipschitz stability.} In addition to permutation equivariance, the \textit{graph-adaptive max nonlinearity} is Lipschitz stable to input perturbations with respect to the infinity norm $\|\cdot\|_\infty$ as stated in the following proposition.
\begin{proposition}[Lipschitz stability] \label{thm:lipschitz}
Let $\ccalG$ be a graph with GSO $\bbS$. Assume that $\bbS$ is normalized by its largest eigenvalue so that its spectral norm $\rho(\bbS)$ is unitary. Let $\bbx$ be a graph signal and let $\tbx$ be a perturbation of $\bbx$. The output of the graph-adaptive max activation function \vskip-.2cm
\begin{equation} 
    [\bbz]_i = \beta \text{ReLU}([\bbx]_i) + \sum_{k=1}^{K} h_{\sigma k} [\text{max}(\bbS^{k}\bbx, \mathcal{N}_i)]_i
\end{equation}
with coefficients $|h_{\sigma k}| \leq C$ is Lipschitz stable to input perturbations in the infinity norm $\|\cdot\|_\infty$. That is, there exists a constant $L_\sigma$ such that \vskip-.2cm
\begin{equation}
\|\tbz-\bbz\|_\infty \leq L_\sigma \|\tbx-\bbx\|_\infty
\end{equation}
where $L_\sigma = |\beta| + KC\max_{k} \|\bbS^{k}\|_\infty$.
\end{proposition}
\begin{proof}
For the proof, we refer the reader to the Appendix.
\end{proof}

Proposition \ref{thm:lipschitz} implies the graph-adaptive max activation is Lipschitz stable \textit{at each node}. Lipschitz stability is crucial to make learning more robust. For instance, in classification problems, a GNN with graph-adaptive max nonlinearities will more likely classify correctly a perturbed signal $\tbx$ than a GNN with non-Lipschitz activation functions. The Lipschitz constant depends on the coefficient $\beta$, the number of filter taps $K$, the weights $h_{\sigma k}$ (through $C$), and the graph (through $\max_{k} \|\bbS^{k}\|_\infty$). While we may not have full control over $\max_{k} \|\bbS^{k}\|_\infty$, $\beta$ and $K$ are design parameters, and so is the maximum value of the coefficients $h_{\sigma k}$. The Lipschitz constant of graph-adaptive max nonlinearities is thus \textit{tunable}. This represents an advantage compared to conventional pointwise activation functions, which are stable but have fixed Lipschitz constants. 

\section{Numerical Experiments}
\label{sec:sims}
\vskip-.3cm
We evaluate the performance of six activation functions that include: ReLU, localized activation functions (max and median) \cite{ruiz18-local}, and our proposed graph-adaptive localized (max and median) and kernel activation functions. Our goal is to highlight the benefits and limitations of the different nonlinearities in applications requiring distributed computations with both synthetic and real data. To train the GCNNs we used the ADAM optimizer with learning rate $10^{-3}$ and forgetting factors $\beta_1 = 0.9$ and $\beta_2 = 0.999$. As the GSO, we employ the adjacency matrix normalized by the maximum eigenvalue. For the graph-adaptive kernel nonlinearity, we set the parameter $\gamma$ in (\ref{eqn:gaussian_kernel}) to $\gamma = 0.1$\footnote{The code can be found at \url{https://github.com/bianca26/graph-adaptive-activation-functions-gnns}.}.

\subsection{Source Localization}

We consider a diffusion process over a graph of $N=40$ nodes divided into $C = 4$ communities. The goal is to determine the source community of a given diffused signal locally at a selected node. The graph is an undirected stochastic block model (SBM) with intra- and inter-community probabilities $p = 0.8$ and $q = 0.1$, respectively. The graph signals are defined as Kronecker deltas $\bbdelta_c \in \reals^N$ centered at a source node $c$ and diffused at a timestamp $t \in [0, 30]$, i.e. $\bbx_t = \bbS^t \bbdelta_c$. We choose as source node $c$ each of the $40$ nodes, thus generating a data set consisting of $1200$ graph signals. We split these samples into training, validation, and test set respectively as $80\%$, $10\%$, and $10\%$. We simulate 10 different graphs and generate 10 different splits per graph. The training and testing are performed for the highest connected node for each community, resulting in four nodes. Training is performed for $400$ epochs with a batch size of $100$ samples. 

Table \ref{tab:sourceLoc_results} shows the classification accuracy for a two layer GCNN with different number of features, $F \in \{2, 4, 8\}$. For the graph-adaptive nonlinearities, we carried out the experiments with resolutions $K=1$ and $K=2$. We only report the results for the better performing filter order, as the rest were comparable to the localized nonlinearities from \cite{ruiz18-local}. We observe both the localized nonlinearities and the proposed graph-adaptive nonlinearities significantly outperform ReLU, with a difference in classification accuracy of at least $14\%$. This result highlights the benefits of accounting for the graph topology during classification. Moreover, the graph-adaptive max and median activation functions outperform their localized versions, confirming the advantage of accounting for further away data-graph coupling. The max nonlinearities achieve a higher accuracy than medians in both the localized and graph-adaptive localized nonlinearities. This result could be caused by the fact that the median will overall smooth the signal, hence undermining some local variations important for classification. Additionally, this could also explain the lower performance of the graph-adaptive kernel nonlinearities compared to the localized nonlinearities, which might be affected by the possible redundancies in the extra information coming from neighbors.

\subsection{Distributed Finite-Time Consensus}

Distributed finite-time consensus aims to achieve consensus among all nodes in finite-time, by accessing only local information at each node. We consider learning the distributed consensus function in a data-driven fashion over an undirected SBM graph with $N=100$ nodes divided into $C = 5$ communities with intra- and inter-community probabilities $p = 0.8$ and $q = 0.1$, respectively. The graph signals are generated from a normal distribution $\mathcal{N}(\bbzero,\bbI)$. We generate $2500$ samples and split them into $80\%$, $10\%$, $10\%$ training, validation, and test sets, respectively. We average the performance across $10$ different graph realizations and $10$ different data splits for each graph. We consider a two layer GCNN with $F = 32$ features per layer followed by a per-node fully connected layer. We employ various number of filter orders $K \in \{20, 25, 30, 35\}$. Training is performed for $400$ epochs with batch size $100$. The evaluation metric is the RMSE.

Figure \ref{fig:ftc_exp1} shows the RMSE as a function of the filter order for the different nonlinearities. All GCNNs achieve a lower RMSE compared with the FIR graph filter. For the lowest order $K=20$, ReLU yields a worse RMSE than the localized and graph-adaptive nonlinearities. Once the filter order increases, and thus the degrees of freedom, adding a parametric nonlinearity seems to be less beneficial because the network has enough degrees of freedom in the filter to model the consensus function. We also experiment with the robustness of the different models to link losses by removing graph edges with different probabilities, following the random edge sampling model of \cite{isufi2017filtering}. For each method, we considered the best performing setup.  From the trained graph $\ccalG$, we randomly removed edges with probabilities in the interval $[0.025, 0.15]$. The results are shown in Figure \ref{fig:ftc_exp2}, averaged across $10$ realizations. Although all models deteriorate when the link losses increase, graph-adaptive nonlinearities handle the stochasticity better. The kernel nonlinearity seems to be the most sensitive as its performance reaches those of the other graph-adaptive alternatives.

\begin{table}[!t]
\centering
\caption{Source Localization Test Accuracy. L.: localized nonlinearities [cf. \ref{eq.localizedactiv}]; G.A.: graph-adaptive nonlinearities [cf. (\ref{eqn:distr_activ_fct}) and (\ref{eqn:distr_kernel_activ_fct})]. Between brackets the filter order $K$ is specified.}
\resizebox{.48\textwidth}{!}{%
{\small
  \begin{tabular}{l | c c c}
\Xhline{4\arrayrulewidth}
Nonlinearity/ $F$ & 2 & 4 & 8 \\
\hline
  \rowcolor{Gray}
ReLU & $\mathbf{47.9 (\pm 12.1)\%}$ & $44.9 (\pm 15.6)\%$  & $47.2 (\pm 15.5)\%$\\
Max L. & $64.5 (\pm 8.0)\%$ & $69.7 (\pm 8.6)\%$ & $\mathbf{72.2 (\pm 7.7)\%}$ \\
  \rowcolor{Gray}
Max G.A. (2) & $64.9 (\pm 7.6)\%$ & $69.2 (\pm 7.0)\%$ & $\mathbf{73.9 (\pm 6.8)\%}$ \\
Median L. & $61.6 (\pm 7.4)\%$ & $65.1 (\pm 8.3)\%$  &  $\mathbf{69.6 (\pm 7.2)\%}$ \\
\rowcolor{Gray}
Median G.A. (2) & $65.4 (\pm 7.5)\%$ & $65.6 (\pm 7.6)\%$  & $\mathbf{71.3 (\pm 7.1)\%}$ \\
Kernel G.A. (1) & $58.6 (\pm 9.5)\%$ & $57.4 (\pm 10.2)\%$ & $\mathbf{61.9 (\pm 10.7)\%}$ \\
\Xhline{4\arrayrulewidth}
\end{tabular}}}
\label{tab:sourceLoc_results}\vskip-4mm
\end{table}

\subsection{Distributed Regression}
\begin{figure*}
     \centering
     \begin{subfigure}[b]{0.22\textwidth}
         \centering
         \includegraphics[width=\textwidth]{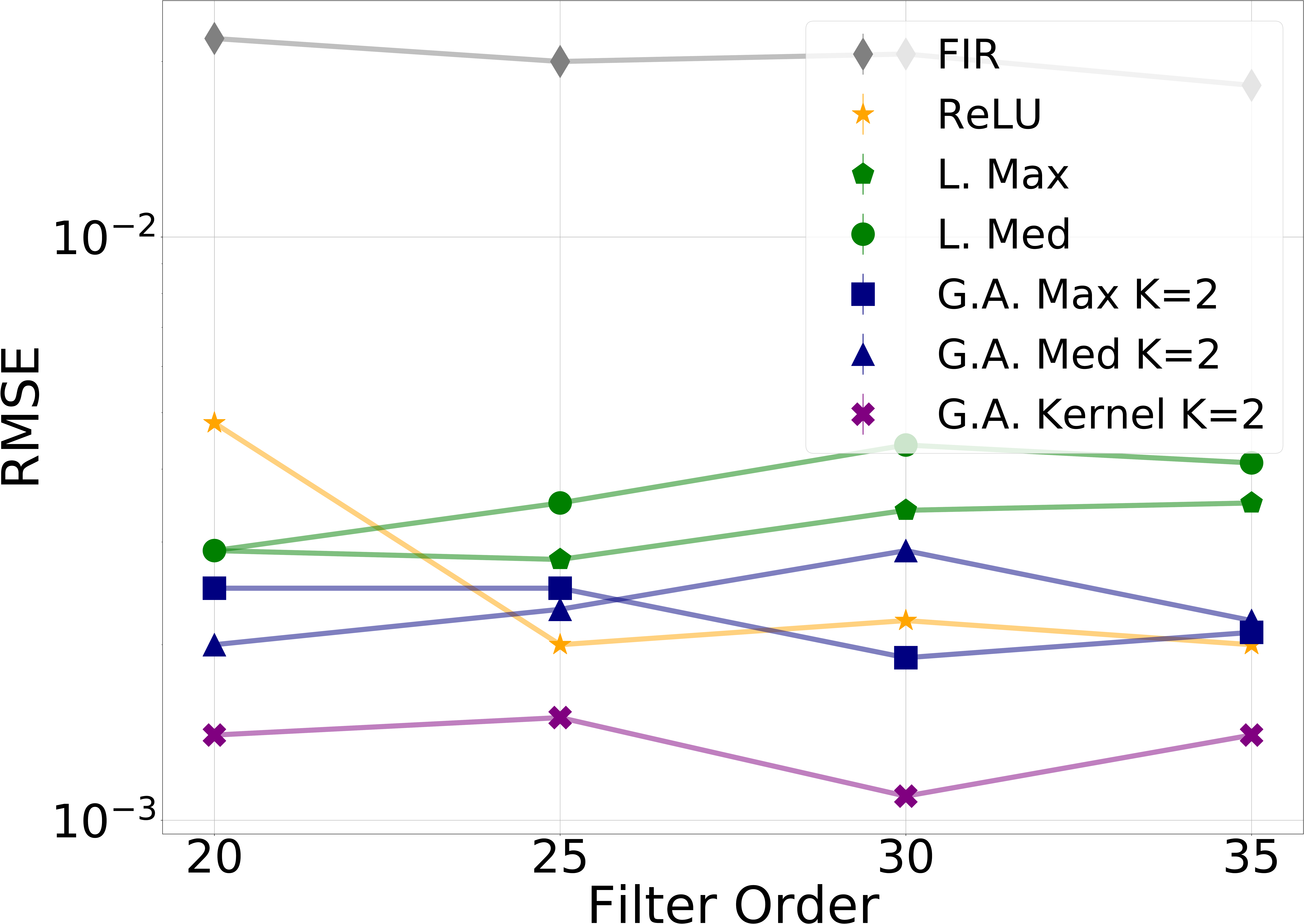}
         \caption{}
         \label{fig:ftc_exp1}
     \end{subfigure}
     \hfill
     \begin{subfigure}[b]{0.22\textwidth}
         \centering
         \includegraphics[width=\textwidth]{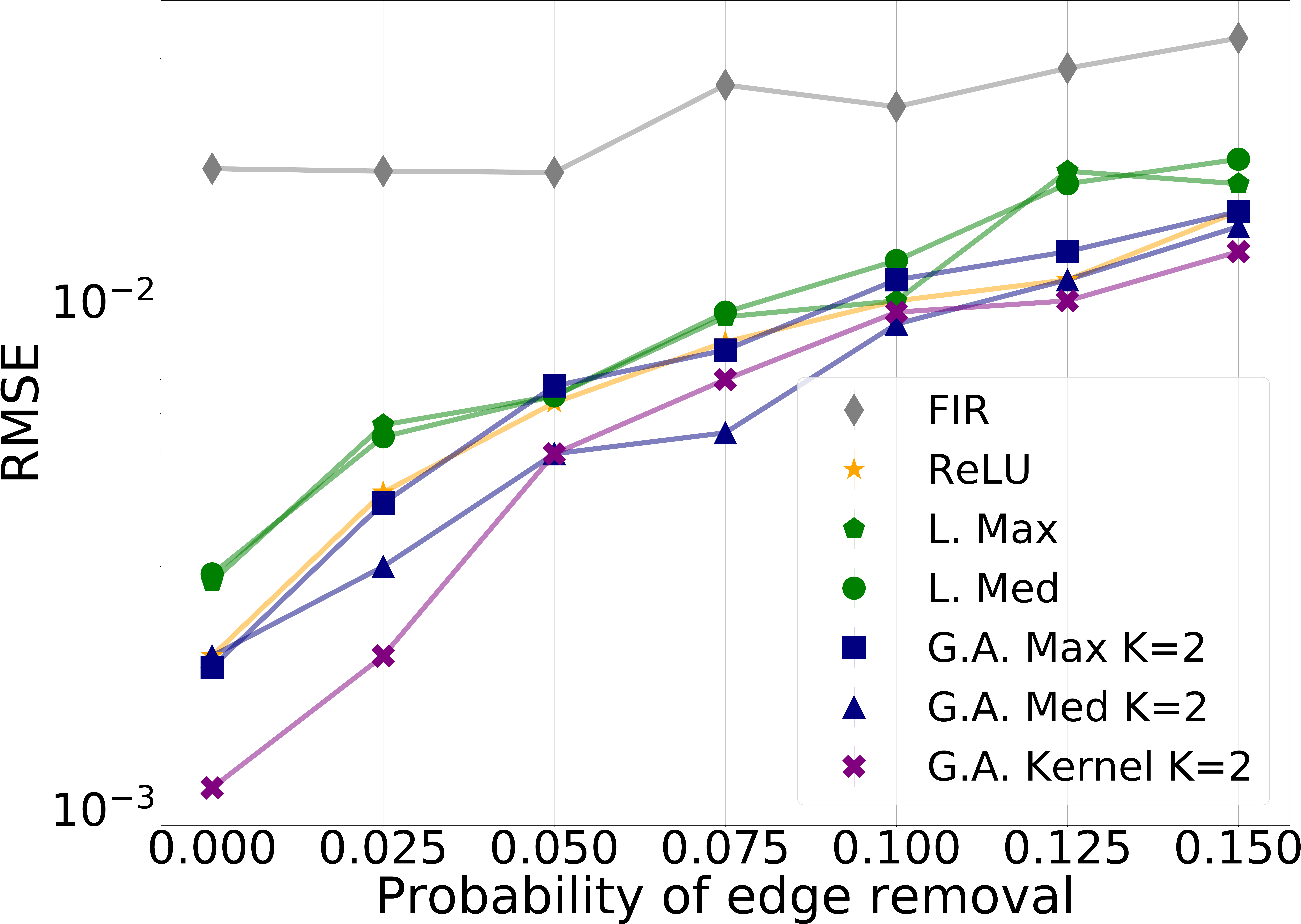}
         \caption{}
         \label{fig:ftc_exp2}
     \end{subfigure}
     \hfill
     \begin{subfigure}[b]{0.22\textwidth}
         \centering
         \includegraphics[width=\textwidth]{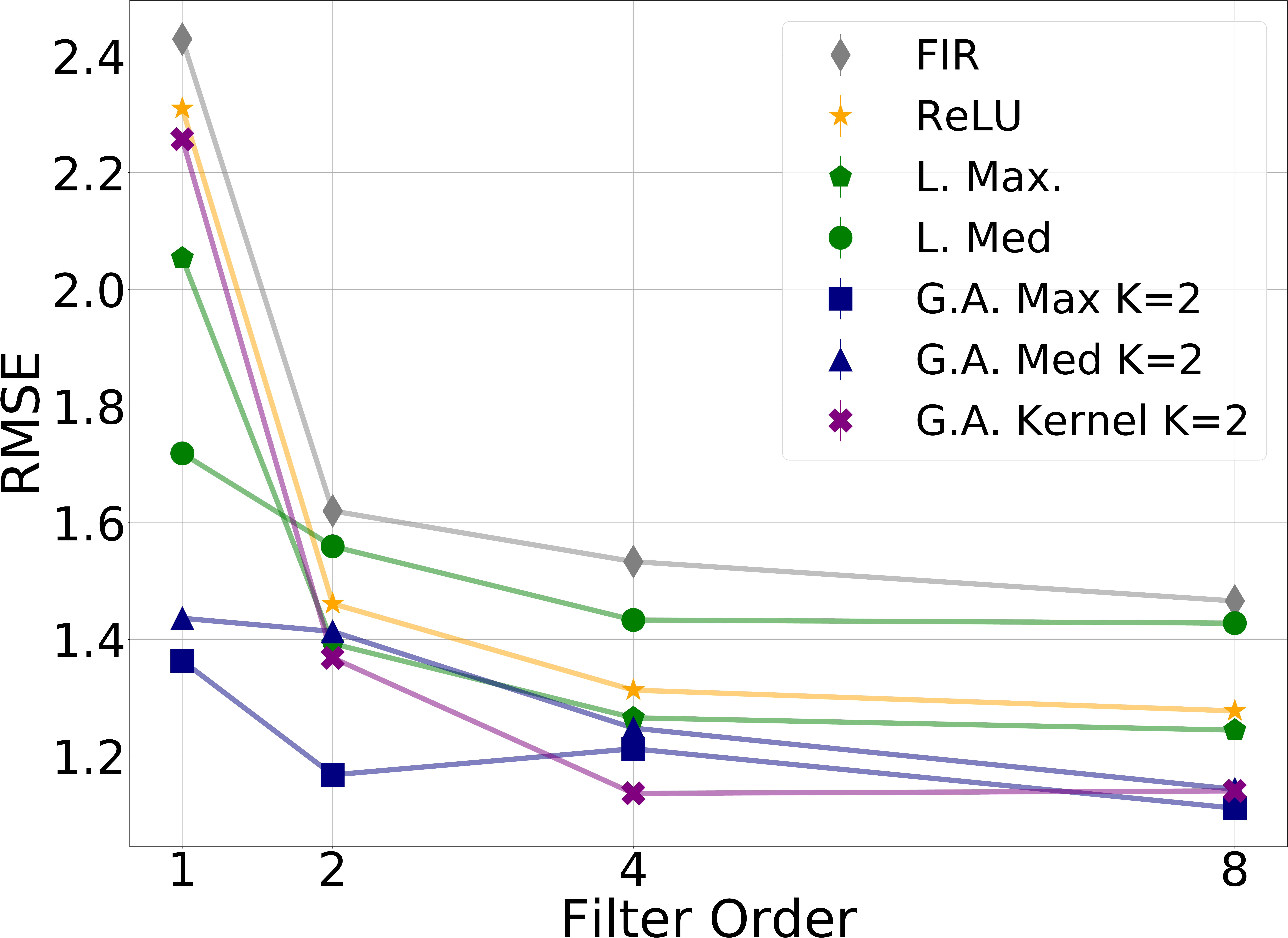}
         \caption{}
         \label{fig:molene_exp1}
     \end{subfigure}
     \hfill
     \begin{subfigure}[b]{0.22\textwidth}
         \centering
         \includegraphics[width=\textwidth]{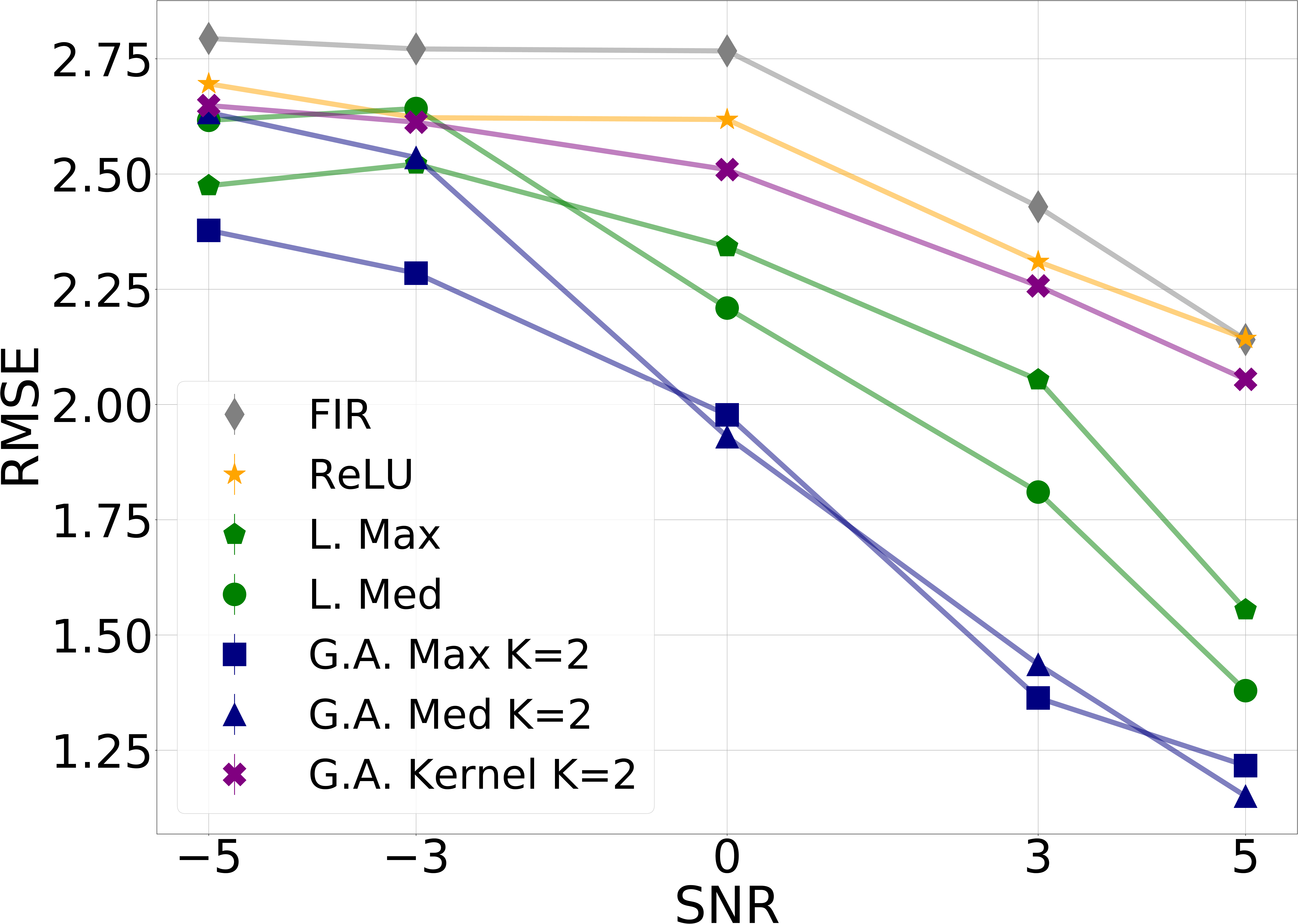}
         \caption{}
         \label{fig:molene_exp2}
     \end{subfigure}
        \caption{(a) Root mean square error (RMSE) of the GCNNs and FIR graph filters for distributed finite-time consensus as a function of filter order. (b) Robustness of the GCNNs and FIR graph filters for distributed finite-time consensus as a function of link-loss probabilities. (c) RMSE of the GCNNs and FIR graph filters for distributed regression as a function of filter order. (d) Robustness of the GCNNs and FIR graph filters for distributed regression as a function of the SNR. L.: localized nonlinearities [cf. \ref{eq.localizedactiv}]; G.A.: graph-adaptive nonlinearities [cf. (\ref{eqn:distr_activ_fct}) and (\ref{eqn:distr_kernel_activ_fct})]. $K$ denotes the filter order.}\vskip-.15cm
        \label{fig:molene_experiments}
\end{figure*}

We perform distributed regression using the Molene dataset, which contains hourly temperature measurements of $N = 32$ stations over $T = 744$ hours recorded in January 2014 in the area of Brest (France). Using the node (station) coordinates, we generate a weighted geometric graph using a ten nearest neighbor approach proposed in \cite{isufi2019forecasting}. We consider as graph signals the measurements taken at different timestamps $t \in T$. Thus, our data set consists of $744$ graph signals. On top of the original signals we add zero-mean noise with a signal-to-noise ratio (SNR) of 3 dB. These noisy signals are split into $80\%$, $10\%$, $10\%$ training, validation, and test sets, respectively. Our goal is to train a GCNN for removing the noise distributively. We employ a GCNN with one layer and a varying number of features $F \in \{1, 2, 4\}$ and filter orders $K \in \{1, 2, 4, 8\}$. We perform the training for 500 epochs with a batch size of 100 samples. We employ RMSE as the evaluation metric. The final results are averaged across 20 different splits of the data set.

Figure \ref{fig:molene_exp1} shows the RSME as a function of the filter order for the different nonlinearities. Across all GCNNs, the best performance was achieved for the highest number of features, four, so we only report the results for this setup. In the other setups, the performances were comparable. All GCNNs perform better than the FIR, but the difference is more significant for the lowest filter order $K=1$, especially in the case of graph-adaptive localized nonlinearities. This finding suggests their applicability in situations where the communication resources are limited. To further address this hypothesis, we experimented with different levels of noise added to the data. For each method, we considered the setup with the lowest filter order $K=1$. The results in Figure \ref{fig:molene_exp2} show that the graph-adaptive and localized nonlinearities outperform or achieve comparable results to ReLU. The general trend shows an increase in performance when the SNR becomes larger, with a more significant increase for the graph-adaptive localized nonlinearities. The performance of the graph-adaptive kernel nonlinearity suffers in this scenario, as it requires higher filter orders compared to the rest. We suggest using higher orders in the latter case to fully exploit the kernel power.

\subsection{Recommender Systems}
We implement a GNN-based recommender system by considering a $U \times M$ rating matrix $\bbR$, containing 100,000 ratings given by $U=943$ users to $M=1682$ movies in the MovieLens 100k dataset \cite{harper16-movielens}. The entries $[\bbR]_{um}$ are the ratings between $1$ and $5$ if user $u$ has rated movie $m$, and $0$ otherwise. We interpret the rows of $\bbR$, i.e., the user rating vectors $\bbr_u$, as graph signals on a $M$-node movie similarity network. The graph signals are split into 90\% as training and 10\% as test set, and the movie similarity network is built by computing pairwise correlations between movie rating vectors (i.e., columns of $\bbR$) containing only ratings from users in the training set. The GNN is trained to predict user ratings to a movie $m$. This is achieved by ``zeroing'' out the ratings to movie $m$ in the input graph signals $\bbr_u$, feeding them to the GNN to generate the rating prediction $[\overline{\bbr_u}]_m$, and minimizing the smooth $\ell_1$ loss $|[\bbr_u]_m-[\overline{\bbr_u}]_m|$. We consider three graph-adaptive GNNs employing the one-hop max, one-hop median, and one-hop kernel graph-adaptive nonlinearities to highlight the impact of immediate neighboring information, hence making the recommendation more localized over items. They are compared with GNNs containing ReLU activations and the one-hop max and median activations from \cite{ruiz18-local}. All GNNs consist of  $L=1$ layer, $F=32$ features using graph convolutional filters banks with $K=5$ filter taps each. We train all GNNs over $30$ epochs and in batches of $5$ for the movies Toy Story, Contact, and Return of the Jedi. The average test RMSEs over five random train-test splits for each movie are reported in Table \ref{tab:recSystems_results}.

\begin{table}[!t]
\centering
\caption{Average test RMSE over five train-test splits for the movies Toy Story, Contact and Return of the Jedi. L.: localized nonlinearities [cf. \ref{eq.localizedactiv}]; G.A.: graph-adaptive nonlinearities [cf. (\ref{eqn:distr_activ_fct}) and (\ref{eqn:distr_kernel_activ_fct})].}
\resizebox{.48\textwidth}{!}{%
{\small
  \begin{tabular}{l | c c c}
\Xhline{4\arrayrulewidth}
Nonlinearity & Toy Story & Contact & Return of the Jedi \\
\hline
  \rowcolor{Gray}
ReLU & ${0.976 (\pm 0.158)}$ & $1.022 (\pm 0.042)$  & $1.018 (\pm 0.177)$\\
Max L. & $0.999 (\pm 0.166)$ & $1.028 (\pm 0.029)$ & ${1.001 (\pm 0.162)}$ \\
  \rowcolor{Gray}
Max G.A. & $\mathbf{0.968 (\pm 0.168)}$ & $\mathbf{1.018 (\pm 0.038)}$ & $\mathbf{0.998 (\pm 0.172)}$ \\
Median L. & $0.987 (\pm 0.156)$ & $1.039 (\pm 0.055)$  &  ${1.020 (\pm 0.177)}$ \\
\rowcolor{Gray}
Median G.A. & $0.989 (\pm 0.160)$ & $1.020 (\pm 0.038)$  & ${1.021 (\pm 0.181)}$ \\
Kernel G.A. & $0.977 (\pm 0.152)$ & $1.021 (\pm 0.037)$ & ${1.014 (\pm 0.177)}$ \\
\Xhline{4\arrayrulewidth}
\end{tabular}}}
\label{tab:recSystems_results}\vskip-4mm
\end{table}
We observe the graph-adaptive max activation function outperforms the other nonlinearities for all three movies. In particular, the graph-adaptive max fares better than both the ReLU and its localized counterpart. The graph-adaptive median also outperforms the localized median for the movie Contact, and achieves comparable performance for the other movies. As for the graph-adaptive kernel activation, it performs similarly to the ReLU and does not provide much of an improvement.


\section{Conclusions}
\label{sec:conclusions}
We proposed a new family of \textit{graph-adaptive} activation functions for GNNs that capture the graph topology while also being \textit{distributable}. These activation functions incorporate the data-topology coupling into all the GNN components by combining nonlinearized features from neighboring nodes with a set of trainable parameters. These parameters adapt the information coming from neighborhoods of different resolutions to the task at hand, hence aiding learning.
The proposed graph-adaptive activation functions preserve permutation equivariance, and the graph-adaptive max activation function is Lipschitz stable to input perturbations.
Graph-adaptive nonlinearities were compared to GCNNs employing localized and pointwise nonlinearities in four different problems based on both synthetic and real-world data, showing an improved performance compared to pointwise and other state-of-the-art localized nonlinearities. Future work will be on two fronts: characterizing the stability of the proposed activation functions to perturbations in the topology and performing learning distributively.

\bigbreak
\centering
\textbf{APPENDIX}
\begin{proof}[\textbf{Proof of Prop. \ref{prop:perm_equiv}}] Let $\bbS' = \bbP^\top \bbS \bbP$ be the graph permutation and $\bbx' = \bbP^\top \bbx$ the permuted signal. From \eqref{eqn:graph_conv_perm_equiv}, the output of the graph convolution is equivariant to the action of $\bbP$. Hence, we only need to prove permutation equivariance of the graph-adaptive activation functions in (\ref{eqn:distr_activ_fct}) and (\ref{eqn:distr_kernel_activ_fct}). We write their output as the signal $\bbz$ with entries
\begin{equation} \label{eqn:perm_equiv1}
    [\bbz]_i = \beta \text{ReLU}([\bbx]_i) + \sum_{k=1}^{K} h_{\sigma k} [g(\bbS^{k}\bbx, \mathcal{N}_i)]_i
\end{equation}
where $g(\cdot, \ccalN_i)$ denotes either a shifted localized operator [cf. Def \ref{def:shifted_loc_op}] or a kernel operator [cf. Def. \ref{def:kernel_op}].
Applying the activation functions in \eqref{eqn:perm_equiv1} to the permuted signal $\bbx'$, we obtain
\begin{equation}
    [\bbz']_i  = \beta \text{ReLU}([\bbP^\top \bbx]_i) + \sum_{k=1}^{K} h_{\sigma k} [g((\bbP^\top \textbf{S} \bbP)^{k} \bbP^\top \bbx, \mathcal{N}_i)]_i \text{.}
\end{equation}
Since the ReLU activation function is pointwise, it is permutation equivariant, i.e. $\text{ReLU}(\bbP^\top \bbx) = \bbP^\top \text{ReLU}(\bbx)$. We then focus on the second term of the sum, where we observe that $(\bbP^\top \bbS \bbP)^{k} = \bbP^\top \bbS \bbP \bbP^\top \bbS \bbP ... \bbP^\top \bbS \bbP = \bbP^\top \bbS^{k} \bbP$,
which implies $(\bbP^\top \bbS \bbP)^{k} \bbP^\top \bbx = \bbP^\top \bbS^{k} \bbx.$ We can rewrite $\bbz'$ as
\begin{equation} \label{eqn:perm_equiv2}
    [\bbz']_i = \beta [\bbP^\top \text{ReLU}(\bbx)]_i + \sum_{k=1}^{K} h_{\sigma k} [g(\bbP^\top \bbS^{k} \bbx, \ccalN_i)]_i \text{.}
\end{equation}
Because function $g(\cdot, \ccalN_i)$ is localized, it acts on the one-hop neighborhoods of each node, which are preserved under node relabelings. Therefore, $g(\cdot, \ccalN_i)$ is permutation equivariant and \eqref{eqn:perm_equiv2} becomes
\begin{align*}
    [\bbz']_i &= \beta [\bbP^\top \text{ReLU}(\bbx)]_i + \sum_{k=1}^{K} h_{\sigma k} [\bbP^\top g(\bbS^{k} \bbx, \ccalN_i)]_i\\ 
    &= [\bbP^\top \beta \text{ReLU}(\bbx)]_i +\bigg[ \bbP^\top\sum_{k=1}^{K} h_{\sigma k} g(\bbS^{k} \bbx, \ccalN_i)\bigg]_i \text{.}
\end{align*}
Therefore $\bbz' = \bbP^{\top} \bbz$ and, hence, GNNs with graph-adaptive activation functions are permutation equivariant.
\end{proof}

\begin{proof}[\textbf{Proof of Prop. \ref{thm:lipschitz}}]
Let $\tbx$ be a perturbed input with $i$th entry $[\tbx]_i = [\bbx]_i + \epsilon_i$. Denoting by $\tbz$ the output obtained by applying the graph-adaptive max activation function to $\tbx$, we can write
\begin{align} \label{eqn:lipschitz_proof1}
\begin{split}
    \|[\tbz]_i-[\bbz]_i\| \leq \left\|\beta \left(\text{ReLU}([\tbx]_i)-\text{ReLU}([\bbx]_i)\right)\right\| \\
   \quad \quad + \left\|\sum_{k=1}^{K} h_{\sigma k} \left([\text{max}(\bbS^{k}\tbx, \mathcal{N}_i)]_i- [\text{max}(\bbS^{k}\bbx, \mathcal{N}_i)]_i\right) \right\|
   \end{split}
\end{align}
which is obtained by grouping terms and applying the triangle inequality. The ReLU activation is Lipschitz stable with constant one \cite{wiatowski2017mathematical}, and so 
\begin{align} \label{eqn:lipschitz_proof2}
\begin{split}
    \left\|\beta \left(\text{ReLU}([\tbx]_i)-\text{ReLU}([\bbx]_i)\right)\right\| \!\leq |\beta|\left\|[\tbx]_i-[\bbx]_i\right\|
    \!= |\beta|\left\|[\tbx-\bbx]_i\right\| \text{.}
    \end{split}
\end{align}
For the second part of the sum in \eqref{eqn:lipschitz_proof1}, we have
\begin{align*}
\begin{split}
    &\left\|\sum_{k=1}^{K} h_{\sigma k} \left([\text{max}(\bbS^{k}\tbx, \mathcal{N}_i)]_i- [\text{max}(\bbS^{k}\bbx, \mathcal{N}_i)]_i\right) \right\| \\
    &\leq \sum_{k=1}^K |h_{\sigma k}| \left\|[\text{max}(\bbS^{k}\tbx, \mathcal{N}_i)]_i- [\text{max}(\bbS^{k}\bbx, \mathcal{N}_i)]_i\right\| 
\end{split}
\end{align*}
which follows from the Cauchy-Schwarz inequality. Observe that, for any two functions $f(\cdot)$ and $g(\cdot)$, we can write the inequality $\max(f) = \max(f-g+g) \leq \max(f-g) + \max(g)$, and so
\begin{align*}
\begin{split}
    &\sum_{k=1}^K |h_{\sigma k}| \left\|[\text{max}(\bbS^{k}\tbx, \mathcal{N}_i)]_i- [\text{max}(\bbS^{k}\bbx, \mathcal{N}_i)]_i\right\| \\
    &\leq \sum_{k=1}^K |h_{\sigma k}| \|[\max(\bbS^{k}(\tbx- \bbx), \mathcal{N}_i)]_i\|\text{.}
    \end{split}
\end{align*}
We proceed by noting that 
\begin{align*}
\begin{split}
\|[\max(\bbS^{k}(\tbx- \bbx), \mathcal{N}_i)]_i\| \leq \|\max_i [\bbS^{k}(\tbx- \bbx)]_i\| \\
\leq \max_i \|[\bbS^{k}(\tbx- \bbx)]_i\| = \|\bbS^{k}(\tbx- \bbx)\|_\infty
\end{split}
\end{align*}
which allows us to write
\begin{align} \label{eqn:lipschitz_proof3}
\begin{split}
    &\left\|\sum_{k=1}^{K} h_{\sigma k} \left([\text{max}(\bbS^{k}\tbx, \mathcal{N}_i)]_i- [\text{max}(\bbS^{k}\bbx, \mathcal{N}_i)]_i\right) \right\| \\
    &\leq \sum_{k=1}^K |h_{\sigma k}|\|\bbS^{k}(\tbx- \bbx)\|_\infty \leq \sum_{k=1}^K |h_{\sigma k}|\|\bbS^{k}\|_\infty \|\tbx- \bbx\|_\infty \\
    &\leq KC \max_k \|\bbS^{k}\|_\infty \|\tbx- \bbx\|_\infty \text{.}
    \end{split}
\end{align}
Putting \eqref{eqn:lipschitz_proof2} and \eqref{eqn:lipschitz_proof3} together, we can write
\begin{align*}
\begin{split}
     \|[\tbz-\bbz]_i\| = \|[\tbz]_i-[\bbz]_i\| \leq |\beta|\left\|[\tbx-\bbx]_i\right\|
    + KC \max_k \|\bbS^{k}\|_\infty  \|\tbx- \bbx\|_\infty \text{.}
\end{split}
\end{align*}
Since this is true for all $i$ and from the definition of $\|\cdot\|_\infty$, we conclude
\begin{align*}
     \|\tbz-\bbz\|_\infty = \leq \left( |\beta| + KC \max_k \|\bbS^{k}\|_\infty  \right)\|\tbx- \bbx\|_\infty
\end{align*}
which completes the proof. Note that $\|\bbS^{k}\|_\infty \geq \rho(\bbS)^{k} = 1$ for all $k$ with $\lim_{k \to \infty} \|\bbS^{k}\|_\infty = \rho(\bbS)^{k} = 1$, so there exists $K_0$ such that, for all $k > K_0$, $\|\bbS^{k}\|_\infty \leq \max_k \|\bbS^{k}\|_\infty $ with $\max_k \|\bbS^{k}\|_\infty  = \|\bbS^{K_0}\|_\infty$.
\end{proof}

\bibliographystyle{IEEEbib}
\bibliography{myIEEEabrv,bib-nonlinear}

\end{document}